\documentclass[13pt]{amsart}
\usepackage[utf8]{inputenc}
\usepackage[T1]{fontenc}
\usepackage{lmodern}
\usepackage{csquotes}
\usepackage{amsmath, amssymb, amsthm, mathtools, bm}
\usepackage{hyperref}
\hypersetup{unicode=true}
\usepackage{enumitem}
\usepackage{hyperref}
\hypersetup{
    colorlinks=true,
    linkcolor=red,
    citecolor=blue,
    urlcolor=blue
}
\usepackage{geometry}
\usepackage{graphicx}
\usepackage{tikz}
\usepackage{cleveref}

\usepackage{setspace}
\theoremstyle{plain}
\newtheorem{theorem}{Theorem}[section]
\newtheorem{lemma}[theorem]{Lemma}

\newtheorem{proposition}[theorem]{Proposition}

\theoremstyle{definition}
\newtheorem{definition}[theorem]{Definition}

\theoremstyle{remark}
\newtheorem{remark}[theorem]{Remark}

\title{The Coase Theorem and Ideal Exchanges}
\author{Daniel Lü}
\address{Department of Philosophy \\ New York University \\ New York, NY 10003 \\ USA}
\email{d.lu@nyu.edu}
\date{\today}
\subjclass[2020]{Primary 91B02, 03E75, 03A10, 91-03.}
\keywords{Coase's Theorem, microeconomics, externalities, Pareto optimality, invariance, convergence, Ronald Coase, Philosophy of Economics.}

\begin{document}

\begin{abstract}
This paper offers a proof of the Coase theorem by formalizing the notion of \textit{ideal exchanges}.  
\end{abstract}

\maketitle

\section{Introduction}

An \textit{externality} can be defined as an effect, often negative, that is imposed on a third party through a private activity. An instance of this problem occurs when non-smokers and smokers are in each other's vicinity; the non-smoker is adversely affected by the smoker's consumption of cigarettes through no fault of their own. The economic and philosophical significance of the existence of externalities lies in the challenge it poses to the classical position that a free-market economy is capable of allocating resources in a socially efficient manner; if the social costs of private consumption are not reflected in the market price of the product in question, then consumers would be incentivized to consume the good at a level that is beyond socially optimal. A comprehensive overview of the economic explanations surrounding externalities and their historical origins can be found in \cite{boudreaux2019externality}.

Prior to Ronald Coase's work, \textit{The Problem of Social Cost}, the conventional view---which Coase himself attributes to Arthur Pigou \cite[p. 1]{coase1960}---held that externalities could be corrected through state intervention: namely, through taxes or subsidies that are designed to deter people from either over-consuming goods associated with negative externalities or under-consuming goods associated with positive externalities \cite{pigou1920}. Coase reoriented the debate from one that relied solely on governmental regulation to one that recognized the potential efficiency of private negotiation in resolving externalities. In particular, Coase reasoned that externalities resulted from ill-defined property rights; provided that property rights are well-defined and each party could negotiate with the other without hindrance, resources would be distributed away from those who value them less and toward those who value them more irrespective of how the rights are initially distributed \cite[p. 8, para. 2]{coase1960}. This proposition, although informal, was characterized as a \textit{theorem} in George Stigler's \textit{Theory of Price}:

\begin{displayquote}
The Coase theorem thus asserts that under perfect competition private and social costs will be equal. It is a more remarkable proposition to us older economists who have believed the opposite for a generation, than it will appear to the younger reader, who was never wrong, here. \cite{stigler1972}
\end{displayquote} 
 
 Stigler's recognition of the potential mathematical standing of what is now widely termed the Coase theorem invites us to examine the economic proposition precisely, even though an exact formulation is not present in Coase's original work. The prevailing interpretation of the \textit{theorem} adopted by this paper can be stated as follows: \textit{if property rights are well-defined and there are no transaction costs, then rational agents would arrive at an optimal distribution of resources of their own accord.} We find this interpretation in Stigler \cite{stigler1972}, Nutter \cite{nutter1968}, Demsetz \cite{demsetz1972}, and Coase himself \cite{coase1960}. Yet, again, we are confronted with the difficulty of constructing a precise interpretation of informal notions such as the property of being ``well-defined'' or constituting a ``transaction cost''. Additionally, we may interpret the notion of ``arriving at an \textit{optimal} outcome'' as a declaration of the existence of some ultimate, optimal distribution such that it is a necessary outcome of every possible initial distribution; alternatively, we may consider that there are multiple ultimate distributions for each initial distribution that fall within the set of outcomes characterized as being optimal.

There have been numerous attempts to formalize the Coase theorem in the literature, and it is essential to understand how the approach in this paper significantly differs from those earlier efforts. In cooperative game-theoretic terms, efficient bargains can be analyzed using the concept of the \textit{core} of an economy with externalities. In fact, Aivazian and Callen (1981) showed that when more than two parties are involved, the bargaining game may have an empty core, meaning there is no stable allocation that cannot be improved upon by some coalition \cite{aivazianCallen1981}. In their example, any agreement among a subset of the parties is unstable: other parties can always offer a different deal, leading to what they describe as an “endless recontracting” process. This result implies that with three or more affected agents, private bargaining can fail to produce a stable efficient outcome – a direct challenge to a broad interpretation of the Coase theorem. Coase himself responded to such arguments by suggesting that these theoretical instabilities might not undermine the practical insight of his theorem; the study of such examples, he argued, had “not led [him] to modify [his] views” \cite{coase1974}. Nevertheless, the empty core problem highlights the need to explicitly account for coalition formation and bargaining rules when formalizing Coase’s result.

Other formalizations have focused on the conditions required for Coasean bargaining to yield a unique outcome. Hurwicz \cite{hurwicz1995}, for example, proved that quasi-linear utility (i.e. the absence of income effects) is both a sufficient and necessary condition for the final allocation of resources to be invariant to the initial assignment of property rights. In other words, if agents have quasi-linear preferences, then any zero-transaction-cost bargaining leads to the same efficient outcome regardless of who starts with the rights – but if preferences are not quasi-linear, so that wealth effects matter, different initial allocations can lead to different Pareto-efficient outcomes. This finding, sometimes termed Coase’s “neutrality” or invariance theorem, shows that the Coasean prediction of unique outcomes rests on restrictive assumptions. Likewise, formal negotiation models in economics, such as Rubinstein’s alternating-offer bargaining model \cite{rubinstein1982}, confirm that two parties with no transaction costs will reach an efficient agreement in equilibrium, but they do not automatically generalize to multiple parties without additional coordination assumptions.

Despite these prior contributions, there remains value in a more explicit logical formalization of the Coase theorem that lays bare the role of each assumption. The approach in this paper differs from earlier work by axiomatizing the bargaining process itself in a set-theoretic framework. Whereas many treatments either assume an equilibrium concept or rely on specific functional forms for utility, here we start from first principles: a finite set of agents, a finite set of resources, and a dynamic rule for voluntary exchanges. This framework allows us to pinpoint precisely why certain interpretations of the Coase theorem can fail. In particular, our model will illustrate, in line with the core literature, that if bargaining is restricted to bilateral trades, the outcome need not be efficient or unique. We then introduce the notion of ideal exchanges – effectively allowing multi-party agreements that improve at least one agent’s lot without hurting others – and show that under this assumption a strong version of the Coase theorem does hold. The additional contribution of this paper is thus a careful analysis of how the bargaining procedure itself affects the validity of Coase’s result, and a demonstration that allowing sufficiently broad cooperation among agents can rescue the theorem.

The remainder of the paper is organized as follows. In Section 2, we discuss the philosophical and methodological motivations for our approach by examining how the Coase theorem is conventionally illustrated and what hidden assumptions these methods entail. This sets the stage for our formal development. In Section 3, we present the formal foundations of our model: the definitions of agents, resources, preferences, and the rules governing exchanges. Section 4 derives some secondary results, showing that two commonly stated versions of the Coase theorem do not hold under the simplest bargaining protocol. Section 5 introduces the concept of ideal exchanges and proves a formal version of the Coase theorem under that assumption. Finally, Section 6 offers philosophical and concluding remarks, reflecting on how these formal results relate to the original spirit of Coase’s insight and the importance of mathematical philosophy in economic theory.

\section{Preliminaries}\label{Preliminaries}

Before moving to the formal model, it is helpful to review how economists typically illustrate the Coase theorem and to highlight why a more foundational approach might be necessary. Contemporary expositions of the Coase theorem are often informal or rely on geometric arguments that, while intuitively appealing, conceal underlying mathematical assumptions. One common illustration uses the tools of indifference curve analysis in an Edgeworth box, as found in many microeconomics textbooks. We summarize that approach here to ground our discussion.

Consider the existence of two types of goods, $x$ and $y$, such that $Q_x$ denotes the quantity of the former and $Q_y$ denotes the quantity of the latter. It is standard to assume that each 2-tuple $(Q_x,Q_y)$ is present in $\mathbb{R}^{+} \times \mathbb{R}^{+}$ or, equivalently, $\mathbb{R}^{2}_{+}$.\footnote{We include $0$ in $\mathbb{R}^{+}$.} We proceed by declaring the existence of a utility function for some agent $A$ such that $\mathcal{U}_{A}: \mathbb{R}^{2}_{+} \to \mathbb{R}^{+}$ to model the notion that it is possible to map pairs of quantities of goods to some level of utility.\footnote{It would seem that \textit{happiness} is a crude synonym for utility, since the latter is often portrayed as a quantifiable notion.} An agent is \textit{indifferent} between two outcomes \textit{iff} they render the same level of utility. In other words, an indifference curve is a set of all the 2-tuples $(Q_x,Q_y) \in \mathbb{R}^{2}_{+}$ such that $\mathcal{U}_{A}(Q_x,Q_y) = c$, where $c$ is a constant in $\mathbb{R}^{+}$. Furthermore, it is standard for an economist to assume that each indifference curve carries important characteristics. Intuitively speaking, it would be peculiar for an indifference curve to contain $(Q_x,Q_y)$ and $(Q_x',Q_y')$ if $Q_x < Q_x'$ and $Q_y < Q_y'$, since that would imply that the utility of the agent remains constant despite an increase in the presence of both goods. It follows, therefore, that all indifference curves must be monotonically decreasing. A more refined assumption would be that indifference curves should model the empirical observation of \textit{diminishing marginal utility}, that is, every additional unit of utility gained from each additional unit of the good decreases, but never goes below zero.\footnote{If it did, then that would imply that an additional unit of the good resulted in a loss of utility, which is not a standard assumption in microeconomic theory.} Thus, indifference curves are often portrayed as being convex to the agent's origin. 

Let us now consider an exchange economy where there are two goods ($x$ and $y$) and two agents ($A$ and $B$). We represent this with the set $\mathbb{E}=[0,Q_x] \times [0, Q_y]$, where $\mathbb{E}$ stands for an \textit{Edgeworth Box}. Every point in the box effectively forms a partition of the goods among both agents. For example, the point $(0,0)_B = (Q_x, Q_y)$ in Figure 1 is a distribution of resources where $A$ is in possession of everything and $B$ is in possession of nothing; the opposite is true at $(0,0)_A = (0,0)$.\\

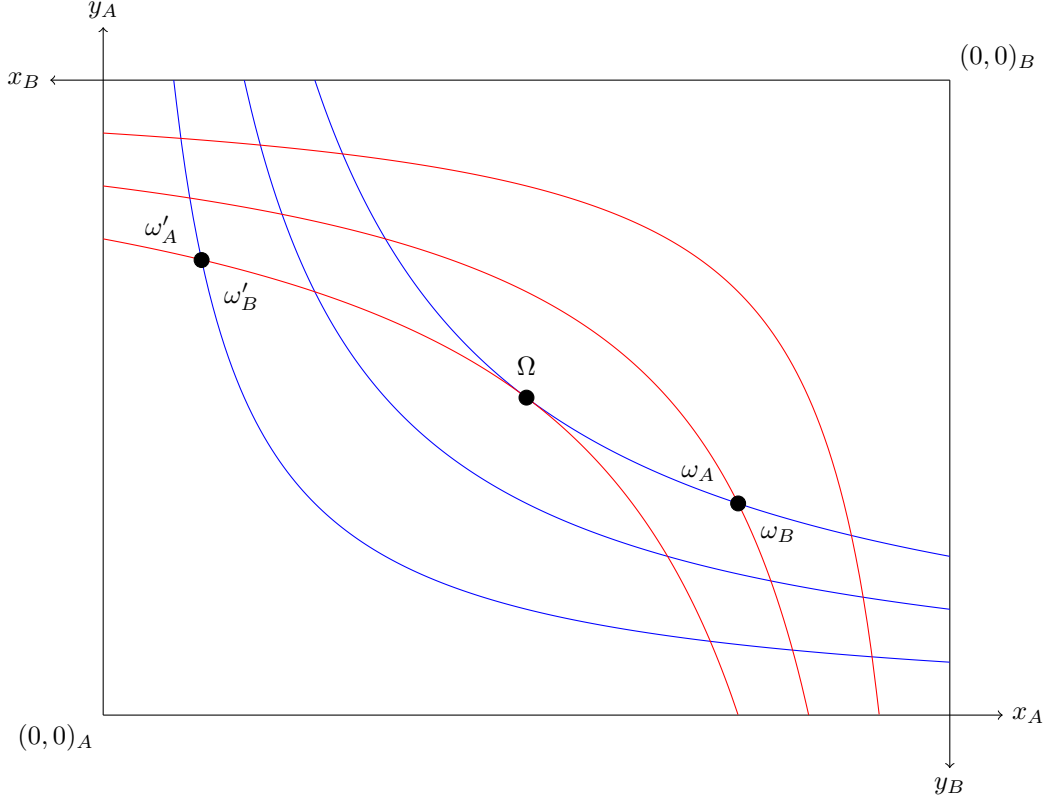
\begin{center}\textbf{Figure 1}: An Edgeworth Box 
\begin{tikzpicture}[scale=1.4]
  \def\TotalX{8}  
  \def\TotalY{6}

  \node at (0,0) [below left] {$(0,0)_{A}$};
  \draw[->] (0,0) -- (\TotalX+0.5,0) node[right] { \(x_A\)};
  \draw[->] (0,0) -- (0,\TotalY+0.5) node[above] { \(y_A\)};
  
  \node at (\TotalX,\TotalY) [above right] {$(0,0)_B$};
  \draw[->] (\TotalX,\TotalY) -- (-0.5,\TotalY) node[left] {\(x_B\)};
  \draw[->] (\TotalX,\TotalY) -- (\TotalX,-0.5) node[below] {\(y_B\)};
  
  \clip (0,0) rectangle (\TotalX,\TotalY);
  
  \foreach \c in {4,8,12} {
    \draw[domain={\c/\TotalY}:\TotalX, samples=100, smooth, variable=\x, blue] 
         plot ({\x},{\c/\x});
    \ifnum\c=8
     
    \fi
  }
  
  \foreach \c in {4,8,12} {
    \draw[domain=0:{\TotalX - \c/\TotalY}, samples=100, smooth, variable=\x, red] 
         plot ({\x},{\TotalY - \c/(\TotalX-\x)});
    \ifnum\c=8
      
    \fi
  }
  
  \filldraw[black] (6,2) circle (2pt)
    node[above left, xshift=-5pt, yshift=5pt] {\(\omega_A\)}
    node[below right, xshift=5pt, yshift=-5pt] {\(\omega_B\)};
  
  \filldraw[black] (0.93,4.3) circle (2pt)
    node[above left, xshift=-5pt, yshift=3 pt] {\(\omega_A'\)}
    node[below right, xshift=5pt, yshift=-5pt] {\(\omega_B'\)};
  
  \filldraw[black](4,3) circle (2pt)
    node[above, xshift=0pt, yshift=5pt] {\(\Omega\)};
  
\end{tikzpicture}
\end{center}

\vspace{5mm}

Consider a family of sets for each agent:

\begin{enumerate}[label=(F\arabic*)]
    \item \emph{Family of Indifference Curves:} \[\forall c \in \mathbb{R}^{+} \, \Big[\big\{(x_A,y_A) \, | \, \mathcal{U}_{A}(x_A,y_A) = c \big\} \in \mathcal{F}_{A} \, \land \, \big\{(Q_x-x_A,Q_y-y_A) \, | \, \mathcal{U}_{B}(Q_x-x_A,Q_y-y_A) = c \big\} \in \mathcal{F}_{B}\Big].\]
\end{enumerate}

Three members of $\mathcal{F}_A$ appear in Figure 1 as blue curves, while three members of $\mathcal{F}_B$ are shown as downward-sloping red curves. A standard definition of the notion of optimality thus follows: if it is not possible to make one agent better off without there being another agent who is made worse off, then the current distribution is Pareto optimal. i.e., 

\begin{enumerate}[label=(P\arabic*)]
    \item \emph{Pareto Optimality of $(x_A,y_A)$:}\label{Pareto Optimality in Edgeworth} \[\forall x_A', y_A' \text{ s.t. } x_A' \neq x_A, \, y_A' \neq y_A, \,  
    \frac{\mathcal{U}_A(x_A',y_A')-\mathcal{U}_A(x_A,y_A)}{\mathcal{U}_B(Q_x-x_A',Q_y-y_A')-\mathcal{U}_B(Q_x-x_A,Q_y-y_A)} \in \mathbb{R}_{<0}\]
\end{enumerate}  

Every Pareto optimal distribution occurs when the indifference curves for both agents are tangent to each other. We can see this by considering the intersection at $\omega_A$ or, equivalently, $\omega_B$. Since $A$ has an excess of good $x$ and $B$ has an excess of good $y$, it would be possible for $A$ to give up just enough of $x$ and receive just enough of $y$ so that she is indifferent between $\omega_A$ and $\Omega$, while $B$'s indifference curve progresses to a point where it is no longer possible to gain any further utility without damaging $A$'s material interests; the opposite is true at $\omega_A'$. If we assume that exchanges can occur \textit{iff} they are mutually beneficial, then the negotiated outcome is strictly inside the lens formed by the two indifference curves. Through continual negotiations, the size of the lens shrinks until it is no longer possible to make a mutually beneficial trade.\footnote{This is a standard description of how negotiations occur. It is notably present in Buchanan and Tullock's foundational work on public choice theory, \textit{The Calculus of Consent: Logical Foundations of Constitutional Democracy} \cite[p. 100]{buchanan1962calculus}. } Hence, the ultimate outcome, although different from $\Omega$, would still be Pareto optimal. 

It is striking that the above account – standard in microeconomic theory – relies on several mathematical constructs that are taken for granted by economists but raise foundational questions. Edgeworth’s original 1881 formulation of utility in \textit{Mathematical Psychics} (MP) introduced real-number utility assignments and continuous indifference curves at a time when the rigorous foundations of the real numbers were still being established \cite[p. 28]{edgeworth1881}. The broad purpose of MP was to apply mathematical techniques to the moral sciences; in doing so, Edgeworth introduced utility measurements by assigning real numbers to represent individuals' levels of satisfaction or happiness, facilitating the comparison and aggregation of utility across different people. Furthermore, MP readily and implicitly utilizes the concept of a continuum—suggesting infinitely divisible goods or utility levels—without being concerned with foundational aspects of real analysis. On a historical note, Dedekind's \textit{Stetigkeit und irrationale Zahlen} was published in 1872, merely 9 years before MP \cite{dedekind1872}; it is therefore likely that Edgeworth was unaware of these developments in the philosophical foundations of mathematics and their effects on the mathematical foundations of economics. 

Similarly, modern expositions blithely assume goods are perfectly divisible and preferences smooth, without considering the logical implications. Suppose agent $A$ derives a utility of exactly $\sqrt{2}$ from consuming the entire unit of $x$.\footnote{$\sqrt{2}$ is an irrational number.} If $B$ owns the good initially, $A$ would be willing to pay up to $\sqrt{2}$ in some numéraire (say dollars) to buy it. But if only rational-dollar amounts can change hands,\footnote{This is not an unreasonable assumption, since currency is usually denominated in rational numbers.} $A$ can never quite pay $\sqrt{2}$ exactly. In a world where we imagine infinitely divisible currency, one might shrug this off; yet it hints that our idealizations hide practical constraints. As another example, imagine two goods with incommensurable values: $A$’s utility for some quantity of good $x$ is a rational multiple of $\sqrt{2}$, and $B$’s utility for some quantity of good $y$ is a rational multiple of $\sqrt{\pi}$. If $A$ has all of good $x$ and $B$ all of good $y$, any rate at which they exchange $x$ for $y$ that equalizes their marginal gains might involve trading an irrational quantity (e.g. $\sqrt{2/\pi}$ of good $y$ for 1 unit of $x$), which is impossible if goods are only measurable in rational units. In such a peculiar economy, the agents might be unable to execute the trade that would make them both better off, even absent transaction costs.

A  final thought experiment: suppose two agents exist in a world where one agent’s endowment is literally the continuum of real numbers in the interval $[0,1]$ and the other’s endowment is the continuum $[2,3]$. At first glance, resources seem scarce (each has a finite interval of the real line). But note that these are infinite (uncountable) sets. Agent 1 could give one specific real number (say $0.5$) to Agent 2. After this “trade,” Agent 1 still has $[0,1] \setminus \{0.5\}$, which is essentially the interval $[0,1]$ with a single point removed – still uncountably infinite, and in fact equinumerous with the original $[0,1]$ (removing one point from a continuum leaves its cardinality $\mathfrak{c}$ unchanged, since $\mathfrak{c} - 1 = \mathfrak{c}$). Likewise, Agent 2’s holdings have increased to $[2,3] \cup \{0.5\}$, which has the same cardinality as $[2,3]$ ($\mathfrak{c} + 1 = \mathfrak{c}$). From each agent’s perspective, nothing has materially changed – each still effectively has a “continuum” of resources – so both agents would be indifferent to this trade. By transferring finitely or countably many singletons from $[0,1]$ to the other agent, after each step the first agent still holds a set of cardinality $\mathfrak{c}$, and, under Lebesgue measure, still total measure $1$. Thus, with cardinality as a welfare proxy, these moves never register a loss for the first agent---even though resources are being transferred. However, no countable sequence of singleton trades can exhaust $[0,1]$; only a transfinite process of length $\mathfrak{c}$ could in principle remove all points. This highlights that cardinality is too coarse a notion for welfare. 

The bizarre scenario violates the spirit of Pareto optimality – we could keep making one party better off while never hurting the other – yet it arises from taking divisibility and infinite sets seriously without additional structure on preferences or the trading process.

The above examples, while contrived, underscore a deeper point: without a sufficiently careful description of how bargaining occurs, the usual informal arguments for the Coase theorem can run into conceptual puzzles or indeterminacies. If agents can trade in arbitrary real quantities and have exotic preferences, what guarantees that they will reach an efficient outcome? Standard economic reasoning assumes away these difficulties by imposing nice properties such as continuity, convexity, or differentiability, and by implicitly presuming that agents negotiate only in mutually beneficial ways. But to truly prove a Coase-like result, one needs to specify the rules of the game.

This motivates an axiomatic approach. We aim to capture the essence of Coasean bargaining in a simple, logical framework, avoiding unnecessary mathematical structures. Our goal is not to reflect every nuance of real bargaining, but to model a world in which rational agents trade resources if and only if it benefits them to do so, and to see what must be assumed for the Coase theorem to hold in such a world. As Paul Cohen remarked about formalizing mathematics, attempts to write down axioms are often meant to codify the rules practitioners implicitly use:

\begin{displayquote}
    The attempts to formalize mathematics and make precise what the axioms are were never thought of as attempts to explain the rules of logic, but rather, to write down these rules and axioms which appeared to correspond to what contemporary mathematicians were using. \cite[p. 1074]{cohen2002}
\end{displayquote}

Here, we attempt something similar for a corner of economic theory: we write down rules that mirror idealized economic assumptions and then explore their consequences. We certainly do not claim these axioms are empirically “true” or that they capture all aspects of real negotiations. Rather, this is a theoretical exercise to determine what follows from particular idealized assumptions about trading behavior. If the Coase theorem is to be regarded as a theorem in the mathematical sense, we must be clear about the assumptions under which it holds and the precise meaning of its conclusion.

\section{Foundations}\label{Assumptions}

We now set up a formal model of a bargaining economy to analyze the claims of the Coase theorem. The model consists of a finite set of agents, a finite set of divisible resources, and a discrete time process during which agents may exchange resources. The rules of exchange will be explicitly defined to capture the idea of voluntary, mutually beneficial trades in a world without transaction costs. This framework is deliberately simplified, consisting of finite sets, no production or uncertainty, so that we can focus on the logical structure of Coasean bargaining.

Let $\mathbb{A}$ be a nonempty finite set of agents (with $|\mathbb{A}| = n$ agents) and let $R$ be a nonempty finite set of resources (with $|R| = m$ distinct divisible goods or items). For convenience, we label the agents as $a_0, a_1, \ldots, a_{n-1}$, and we will often refer to an agent $a_i \in \mathbb{A}$ by index $i$. A distribution of resources at any given time $t$ is a specification of which agent holds each resource. Formally, we represent a distribution as a function $\mathcal{I}_t: \mathbb{A} \to \mathcal{P}(R)$, where $\mathcal{P}(R)$ denotes the power set of $R$ (the set of all subsets of $R$). We impose two natural conditions on any distribution $\mathcal{I}_t$:

\begin{enumerate}[label=(A\arabic*)]
    \item \emph{Well-defined Ownership Rights:} \label{WellDefinedOwnership} \[ \forall \, t \in \mathbb{N} \, \forall \, i \neq j \in [0, n-1] \, \big( \mathcal{I}_t(a_i) \cap \mathcal{I}_t(a_j) = \varnothing \big). \]
    \item \emph{Absence of Unclaimed Resources:} \label{AbsenceOfUnclaimed} \[\forall \, t \in \mathbb{N} \, \bigg[\bigcup_{i=0}^{n-1} \mathcal{I}_t(a_i) = R\bigg]. \]
\end{enumerate}  
\begin{remark}
    These assumptions allow $R$ to be partitioned among $n$ agents. Furthermore, they effectively ensure that no transaction costs are present, since all resources are always being possessed by some agent.
\end{remark}

Each agent has preferences over possible bundles of resources. We represent the strict preference relation of agent $a_i$ by a set $W(a_i)$ consisting of ordered pairs of bundles. If $(A, B) \in W(a_i)$, it means that agent $a_i$ strictly prefers bundle $B$ over bundle $A$. Formally, let $W: \mathbb{A} \to \mathcal{P}(\mathcal{P}(R) \times \mathcal{P}(R))$ be a function assigning to each agent $a_i$ a set $W(a_i)$ of ordered pairs $(A, B)$ with $A, B \subseteq R$. We assume each $W(a_i)$ satisfies the standard rationality properties of a strict total order: 

\begin{enumerate}[label=(B\arabic*)]
    \item \emph{Asymmetry:} \label{asymmetry } 
    \[ \forall \, i \in [0, n-1] \, \forall \, A, B \in \mathcal{P}(R) \, \big[(A, B) \in W(a_i) \implies (B, A) \not\in W(a_i)\big]. \]
    \item \emph{Transitivity:} \label{transitivity} 
    \[ \forall \, i \in [0, n-1] \, \forall \, A, B, C \in \mathcal{P}(R) \, \big[(A, B) \in W(a_i) \, \land \, (B, C) \in W(a_i) \implies (A, C) \in W(a_i)\big]. \]
    \item \emph{Completeness:} \label{Completeness} 
    \[\forall \, i \in [0, n-1] \, \forall \, A, B \in \mathcal{P}(R) \, \big[(A, B) \in W(a_i) \, \lor \, (B, A) \in W(a_i) \big]. \]
\end{enumerate}

In simple terms, each agent’s preference $W(a_i)$ is a strict total order over the set of all possible bundles of resources. We can think of each $W(a_i)$ as generated by a utility function $u_i: \mathcal{P}(R) \to \mathbb{R}$ that assigns a numerical utility to each bundle – then $(A,B)\in W(a_i)$ if and only if $u_i(A) < u_i(B)$. However, we \textit{do not} require the existence of a numeric utility representation, as the axioms \ref{asymmetry }–\ref{Completeness} are just the ordinal preference assumptions that ensure a well-defined ranking.

\begin{remark}

   The asymmetry of one's preferences rests on the position that a rational agent cannot be indifferent between materially distinct outcomes. Transitivity and completeness are standard assumptions in economic theory; the latter ensures that the agent is capable of making meaningful comparisons across all possible options and the former ensures that it is possible to make meaningful inferences about one's preferences. If completeness were rejected, then there could exist some bundle $A'$ such that $a_q$ is silent on its existence. This is not rational since it does not prescribe a course of action if $A'$ were offered to the agent. If transitivity were rejected, then it is possible for an agent to prefer $B$ over $A$ and $C$ over $B$ without preferring $C$ over $A$. Suppose an offer were made to such an agent to exchange their $A$ for $C$. The offer would be rejected on the grounds that the agent does not have an explicit preference for $C$ over $A$. But the agent admits that $C$ is a superior material outcome relative to $A$ by accepting an exchange with $B$ and then with $C$. It follows that transitivity is an essential and defining quality of rational conduct.\footnote{This style of reasoning is similar to a Dutch Book argument \cite{sep-dutch-book}.} 
\end{remark}

We now consider two inductive features of \( \mathcal{I}_t \) that hold for all $t \in \mathbb{N}$:  

\begin{enumerate}[label=(C\arabic*)]
    \item \emph{Double Coincidence of Wants:} \label{DoubleCoincidence} 
    \[
     \exists i \neq j \in [0, n-1]  \exists R_1, R_2 \in \mathcal{P}(R) \Big[ \big[ \big( (R_1, R_2) \in W(a_i) \land  (R_2, R_1) \in W(a_j) \big) \land \big(R_1 \subseteq \mathcal{I}_t(a_i) \land  R_2 \subseteq \mathcal{I}_t(a_j) \big) \big]
    \]
    \[
    \land \big[ \big(\mathcal{I}_{t}(a_i),\big(\mathcal{I}_{t}(a_i) \setminus R_1 \big) \cup R_2 \big) \in W(a_i) \land \big(\mathcal{I}_{t}(a_j),\big(\mathcal{I}_{t}(a_j) \setminus R_2 \big) \cup R_1 \big) \in W(a_j)\big]\Big]    
    \]
    \[
    \implies \exists \, i \neq j \in [0, n-1] \, \exists R_1, R_2 \in \mathcal{P}(R) \Big( \mathcal{I}_{t+1}(a_i) = \big(\mathcal{I}_{t}(a_i) \setminus R_1 \big) \cup R_2  \, \land \,  \mathcal{I}_{t+1}(a_j) = \big(\mathcal{I}_{t}(a_j) \setminus R_2 \big) \cup R_1 \Big).
    \]

    \item \emph{Stagnate in the Absence of Mutually Beneficial Trades:} \label{Stagnation} 
    \[
     \forall i \neq j \in [0, n-1] \forall R_1, R_2 \in \mathcal{P}(R) \Big[\big[ \big( (R_1, R_2) \not\in W(a_i) \lor (R_2, R_1) \not\in W(a_j) \big) \lor 
    \]
    \[
    \big(R_1 \not\subseteq \mathcal{I}_t(a_i) \lor \, R_2 \not\subseteq \mathcal{I}_t(a_j)\big) \big] \lor \big[ \big(\mathcal{I}_{t}(a_i),\big(\mathcal{I}_{t}(a_i) \setminus R_1 \big) \cup R_2 \big) \not\in W(a_i) \lor 
    \]
    \[\big(\mathcal{I}_{t}(a_j),\big(\mathcal{I}_{t}(a_j) \setminus R_2 \big) \cup R_1 \big) \not\in W(a_j)\big] \Big] \implies \forall i \in [0, n-1] \big( \mathcal{I}_{t+1}(a_i) = \mathcal{I}_{t}(a_i) \big).
    \]
\end{enumerate}  

\begin{remark}
    In simpler terms, \ref{DoubleCoincidence} says “if two agents have what each other wants and both would be happier after swapping, then they will swap,” and \ref{Stagnation} says “if no such swap exists, nothing changes.” These axioms formalize the idea of frictionless bargaining where any opportunity for a Pareto-improving bilateral trade is immediately seized. They also implicitly assume no transaction costs: exchanges do not consume any resources or utility beyond the goods being swapped, and no goods vanish in the process. Thus, if a beneficial trade is possible, it happens and costs nothing; if this is not possible, the status quo persists.
We emphasize that conditions \ref{DoubleCoincidence} and \ref{Stagnation} are inductive rules governing the evolution of allocations over time. Starting from $\mathfrak{I}_0$, \ref{DoubleCoincidence} and \ref{Stagnation} determine $\mathfrak{I}_1$, which in turn provides the basis to find trades for $\mathfrak{I}_2$, and so on. Because the set of possible distributions is finite, as there are only finitely many ways to assign $m$ goods to $n$ people, this process cannot continue indefinitely without repeating a prior distribution. We will shortly prove that the process indeed converges to a stable allocation in a finite number of steps.

\end{remark}

Lastly, we define more carefully what it means for a distribution of resources to be Pareto optimal. Intuitively, an allocation is Pareto optimal if there is no alternative allocation that would make some agent better off without making any other agent worse off. In our formal setting, this means: a distribution $\mathfrak{I}$ is Pareto optimal if whenever there is at least one agent who strictly prefers some other distribution $\mathfrak{K}$ to $\mathfrak{I}$, then there must be at least one different agent who strictly prefers $\mathfrak{I}$ to $\mathfrak{K}$. Thus, you cannot move to $\mathfrak{K}$ without hurting someone, if $\mathfrak{I}$ is Pareto optimal.

\begin{definition}[Distribution]
    A distribution of resources is an $n$-tuple of the form $\big(\mathcal{I}_t(a_i)\big)_{i=0}^{n-1}$ in $\mathcal{P}(R)^n$ such that it satisfies all preceding conditions. 
\end{definition}

\begin{definition}[Projection Function]
    If $\mathfrak{K}$ is a distribution where $\mathfrak{K} = \big(\mathcal{I}'_t(a_i)\big)_{i=0}^{n-1}$, then there exists a function $\hat{\mathfrak{K}}: \mathbb{A} \to \mathcal{P}(R)$ such that $\forall a_q \in \mathbb{A}[\hat{\mathfrak{K}}(a_q) = \alpha \, \Longleftrightarrow \, \mathcal{I}'_t(a_q) = \alpha ]$. 
\end{definition}

\begin{definition}[Strict Distributional Preference]
    Given an agent $a_q$ and the distributions $\mathfrak{I}=\big(\mathcal{I}_t(a_i)\big)_{i=0}^{n-1}$, $\mathfrak{K}=\big(\mathcal{I'}_t(a_i)\big)_{i=0}^{n-1}$, we say that the agent strictly prefers $\mathfrak{K}$ over $\mathfrak{I}$ (i.e., $\hat{\mathfrak{I}}(a_q) \prec \hat{\mathfrak{K}}(a_q)$) \textit{iff} $\big(\mathcal{I}_t(a_q), \mathcal{I'}_t(a_q)\big) \in W(a_q).$
\end{definition}

\begin{definition}[Pareto Optimality]
Consider the distribution $\mathfrak{I}$. We say that $\mathfrak{I}$ is Pareto optimal \textit{iff} for all alternative distributions $\mathfrak{K}$, if there exists an agent who strictly prefers $\mathfrak{K}$ over $\mathfrak{I}$, then there must be some other agent for whom $\mathfrak{I}$ is strictly preferred over $\mathfrak{K}$, i.e., 
\[\forall \, \mathfrak{K} \neq \mathfrak{I} \in \mathcal{P}(R)^n \, \bigg[\, \exists \, a_q \in \mathbb{A} \,\big[\hat{\mathfrak{I}}(a_q) \prec \hat{\mathfrak{K}}(a_q)\big] \implies \exists \, a_z \neq a_q \in \mathbb{A} \, \big[\hat{\mathfrak{K}}(a_z) \prec \hat{\mathfrak{I}}(a_z)\big]\bigg].
\]

\end{definition}

\section{Secondary Results}

Given the model above, we can now ask: will the agents necessarily reach a Pareto optimal allocation? Is the final allocation unique? How does it depend on the initial allocation? In this section, we derive some general results about the trading process under assumptions \ref{WellDefinedOwnership}–\ref{Stagnation}. We show that the process always converges to a stable allocation in finitely many steps. However, two commonly expected “Coasean” properties do not hold in this basic setting: (1) the final outcome can depend on the initial allocation, and (2) the final outcome need not be Pareto optimal. These results will motivate the refinement of our model in the next section. We begin by working within a fixed economy $\mathcal{E}$ on which every $\mathcal{I}$ operates, where $\mathcal{E}$ is defined as the 3-tuple $(\mathbb{A},R,W)$. 

\begin{lemma}[Convergence Lemma] \label{ConvergenceLemma}
    \[
    \forall \, \big( \mathcal{I}_0(a_i) \big)_{i=0}^{n-1} \in \mathcal{P}(R)^n \, \exists \, t \in \mathbb{N} \, \forall \, k \in \mathbb{N} \Big[ \big( \mathcal{I}_t(a_i) \big)_{i=0}^{n-1} = \big( \mathcal{I}_{t+k}(a_i) \big)_{i=0}^{n-1} \Big].
    \]
\end{lemma}  

\begin{proof}
    Suppose, to the contrary, that there exists some initial distribution $\big( \mathcal{I}_0(a_i) \big)_{i=0}^{n-1} \in \mathcal{P}(R)^n$ such that for all $t \in \mathbb{N}$, there is a $k \in \mathbb{N}$ where $\big( \mathcal{I}_t(a_i) \big)_{i=0}^{n-1} \neq \big( \mathcal{I}_{t+k}(a_i) \big)_{i=0}^{n-1}$. In other words, every ``stable'' distribution that follows from this initial distribution is temporary. There cannot be a case where a distribution is constant for at least two instances before undergoing a non-trivial change; suppose, to the contrary, that such a phenomenon occurred, then the first distribution at $t$ is either one where no agent is willing to trade with another agent, or where no agent is able to trade with another agent. If a break in stagnation occurs at $t+k$, then both conditions must be satisfied at $t+k-1$, but this is contradictory because the distribution at $t+k-1$ descended from $t$. It follows that $\forall t \in \mathbb{N} \Big[ \big( \mathcal{I}_t(a_i) \big)_{i=0}^{n-1} \neq \big( \mathcal{I}_{t+1}(a_i) \big)_{i=0}^{n-1} \Big]$. By \ref{asymmetry }, no exchange can be reversed. Therefore, for indefinitely distinct trades, either the bundles are always different or the agents are always different. Since there are only finitely many resources ($m$) and finitely many agents ($n$), there are finitely many possible distributions ($n^m$) that dually comply with \ref{WellDefinedOwnership} and \ref{AbsenceOfUnclaimed}; hence, the existence of at least one cycle is guaranteed when $t = n^m$.\footnote{This is due to the pigeonhole principle.} By \ref{transitivity}, every cycle is reducible to a reversal of one's preferences, thereby contradicting asymmetry.  It follows that such an initial distribution cannot exist.   
\end{proof}

\begin{remark}
    It would seem that even in the presence of transaction costs, convergence is guaranteed due to the asymmetric nature of the preferences of finite agents over finite resources. 
\end{remark}

\begin{proposition}[Invariance of the Ultimate Outcome]
\label{InvarianceProp}
\begin{align*}
&\forall\,\big(\mathcal{I}_0(a_i)\big)_{i=0}^{n-1} \neq \big(\mathcal{I}'_0(a_i)\big)_{i=0}^{n-1}\in \mathcal{P}(R)^n \, \exists\,\big(\alpha_i\big)_{i=0}^{n-1}\in \mathcal{P}(R)^n \,\\
&\quad  \exists t,t'\in\mathbb{N}\;\;
\forall k,k'\in\mathbb{N} \, \Big[
\big(\mathcal{I}_{t+k}(a_i)\big)_{i=0}^{n-1}
=
\big(\alpha_i\big)_{i=0}^{n-1}=\big(\mathcal{I}_{t'+k'}(a_i)\big)_{i=0}^{n-1}\big].
\end{align*}
\end{proposition}

\begin{theorem}
Proposition \ref{InvarianceProp} is false. 
\end{theorem}

\begin{proof}
    Having previously established Lemma \ref{ConvergenceLemma}, it suffices to construct some pair of distinct initial allocations in $\mathcal{P}(R)^n$ such that they converge to distinct final allocations. Consider the set of initial allocations in $\mathcal{P}\big(\{x\}\big)^{n}$ such that $\forall i \in [0, n-1]\Big(W(a_i) = \Big\{\big(\{\}, \{x\}\big)\Big\}\Big)$. Consider an initial allocation where $\exists i \in [0, n-1]\big(\mathcal{I}_0(a_i)=\{x\}\big)$ and $\forall j \neq i \in [0,n-1]\big(\mathcal{I}_0(a_j)=\{\}\big)$. By \ref{Stagnation} and \ref{ConvergenceLemma}, the distribution is immediately stable and permanent. Now, consider a different initial allocation where $\exists k \in [0, n-1]\big(\mathcal{I}_0(a_k)=\{x\}\big)$, $k \neq i$, and $\forall j \neq k \in [0,n-1]\big(\mathcal{I}_0(a_j)=\{\}\big)$. This distribution is also immediately stable and permanent yet it is different from the ultimate distribution where, in lieu of $k$, agent $i$ was in possession of resource $x$. In other words, there exists a pair of cases where the initial distribution produces a non-trivial effect on the ultimate distribution. 
\end{proof}

\begin{remark}
    This simple example shows that, without additional assumptions, such as the ability to make side payments or quasi-linear utilities to ensure transferable utility, the Coasean outcome need not be unique. The initial property rights allocation can matter for who ends up with the good, even though in both cases the outcome is Pareto optimal. The example may seem trivial, but it underscores that Coase’s theorem in its strongest form – that the same efficient outcome will occur regardless of how rights are initially assigned – is not guaranteed in our model except under special conditions. Indeed, Hurwicz’s result mentioned earlier formalizes exactly such conditions.
\end{remark}

Next, we examine whether the process described will necessarily reach a Pareto optimal allocation. One might assume that because every trade we allow is Pareto-improving, the process of making such trades until none are left must yield a Pareto-optimal outcome. However, the following proposition and theorem show that this is not always the case: the process can get stuck in a stable allocation that is not Pareto optimal, due to the bilateral nature of trades we assumed. In other words, there may exist allocations where no pair of agents can improve through trade, yet a group of three or more of agents could jointly redistribute resources in a way that benefits at least one of them without hurting anyone. Such an allocation would be a stable endpoint of bilateral trades but not Pareto efficient. 

\begin{proposition}\label{CoaseTheoremv1}
    Every ultimate distribution is Pareto optimal. 
\end{proposition}

\begin{theorem}\label{PrimaryDisproof}
    Proposition \ref{CoaseTheoremv1} is false. 
\end{theorem}

\begin{proof}
    Firstly, observe that \ref{asymmetry }–\ref{Completeness} characterize a linear ordering of all elements in $\mathcal{P}(R)$. For each individual agent $a_q$, we reduce these preferences to a $2^m$-tuple through $\mathfrak{P}: \mathbb{A} \to \mathcal{P}(R)^{|\mathcal{P}(R)|}$, where $\mathfrak{P}(a_q)= \big(\alpha_i\big)_{i=0}^{2^m -1}$, such that the following holds: \[ \forall a_q \in \mathbb{A} \bigg[ \exists \, (A,B) \in W(a_q) \Longleftrightarrow \exists \, i \neq j \in [0, 2^m-1] \big[(\alpha_i = A \, \land \, \alpha_j = B)\, \land \, (i < j) \big] \bigg].\]
    
    We may now construct an explicit counterexample to \ref{CoaseTheoremv1}. Suppose that $\mathbb{A}=\{a_0,a_1,a_2\}$, $R=\{x,y,z\}$, and the following statements are satisfied:
\begin{enumerate}
    \item $\mathfrak{P}(a_0)= \big( \{\}, \{z\}, \{x\}, \{y\}, \{z,x\}, \{z,y\}, \{x,y\}, \{x,y,z\} \big)$.
    \item $\mathfrak{P}(a_1)= \big( \{\}, \{y\}, \{z\}, \{x\}, \{y,z\}, \{y,x\}, \{z,x\}, \{x,y,z\} \big)$.
    \item $\mathfrak{P}(a_2)= \big( \{\}, \{x\}, \{y\}, \{z\}, \{x,y\}, \{x,z\}, \{y,z\}, \{x,y,z\} \big)$.
    \item $\big(\mathcal{I}_0(a_0) = \{x\} \, \land \, \mathcal{I}_0(a_1) = \{z\}\big) \, \land \, \mathcal{I}_0(a_2) = \{y\}.$
    \item $\mathfrak{K}=\big(\mathcal{I'}_t(a_i)\big)_{i=0}^{2}.$
\end{enumerate}

Notice that no agent has an incentive to trade. Agent $a_0$ would be willing to exchange her $x$ for $y$, but $a_2$ is not willing to trade her $y$ for $x$; instead, she would prefer to exchange it for $z$. Similarly, $a_1$ is unwilling to exchange her $z$ for $y$ and would rather trade it for $x$. However, $a_0$ is also unwilling to exchange her $x$ for $z$. By \ref{Stagnation} and \ref{ConvergenceLemma}, it follows that this distribution is ultimate. Consider an alternative distribution $\mathfrak{K}$ where $\mathfrak{K}= \big(\{y\}, \{x\}, \{z\}\big)$. Since $a_0$ now has $y$ as opposed to $x$, she strictly prefers $\mathfrak{K}$ over the initial distribution. This is also true for $a_1$ and $a_2$. It follows that since no agent is made worse off in this alternative distribution $\mathfrak{K}$, the initial distribution is not Pareto optimal. 
\end{proof}

\begin{remark}
    The counterexample highlights another key insight: the structure of the bargaining process matters crucially for efficiency. In that example, there was nothing in our rules that allowed $a_0$, $a_1$, and $a_2$ to recognize and execute the three-way swap ${x \to a_1,\; y \to a_0,\; z \to a_2}$ all at once. The Coase theorem’s informal logic assumes that any and all opportunities for mutually beneficial rearrangements will be exploited, but if agents can only trade bilaterally, some opportunities which involve larger coalitions might not be accessible.  This suggests that if we allow sufficiently many agents to cooperate simultaneously, we might restore efficiency. 
\end{remark}

\section{Primary Result}

The failure of bilateral bargaining to guarantee Pareto optimality leads us to refine our model. We introduce the concept of an ideal exchange, which is a trade possibly involving multiple agents and resources, chosen specifically because it makes at least one agent better off while making no one worse off. In essence, instead of requiring trades to be strictly beneficial to each participant in isolation (as in \ref{DoubleCoincidence}), we will now allow a trade if it benefits someone and harms no one. Such an exchange could, for example, involve three parties transferring goods in a cycle, as in the previous counterexample – something no two of them would do alone, but which all three together can agree on because it benefits each of them or at least leaves them no worse. One might think of this as agents being farsighted or capable of forming coalitions to achieve a Pareto improvement, even if any subset of that coalition wouldn’t move unilaterally.

Formally, we modify the trading rule \ref{DoubleCoincidence} to a new rule \ref{D1} that allows any coalition of agents to execute an exchange that results in a Pareto better distribution. We also adjust \ref{Stagnation} to \ref{D2} to define stagnation when no Pareto-improving exchange exists.

\begin{enumerate}[label=(D\arabic*)]
    \item \emph{Ideal Exchanges:}\label{D1}  
    Suppose that $\forall t \in \mathbb{N} \Big[\mathfrak{I}_t = \big( \mathcal{I}_t(a_i) \big)_{i=0}^{n-1}\Big].$ An ideal exchange is a situation where
\[\exists \, \mathfrak{K} \neq \mathfrak{I}_t \in \mathcal{P}(R)^n \, \bigg[\, \exists \, a_q \in \mathbb{A} \,\big[\hat{\mathfrak{I}}_{t}(a_q) \prec \hat{\mathfrak{K}}(a_q)\big] \, \land \, \forall \, a_z \neq a_q \in \mathbb{A} \, \big[\hat{\mathfrak{K}}(a_z) \not\prec \hat{\mathfrak{I}}_{t}(a_z)\big]\bigg] \implies \big(\mathfrak{I}_{t+1} = \mathfrak{K}\big).
\]
    
    \item \emph{Stagnate if Pareto Optimality is Reached:} \label{D2}
   \[\forall \, \mathfrak{K} \neq \mathfrak{I}_t \in \mathcal{P}(R)^n \, \bigg[\, \exists \, a_q \in \mathbb{A} \,\big[\hat{\mathfrak{I}}_{t}(a_q) \prec \hat{\mathfrak{K}}(a_q)\big] \implies \exists \, a_z \neq a_q \in \mathbb{A} \, \big[\hat{\mathfrak{K}}(a_z) \prec \hat{\mathfrak{I}}_{t}(a_z)\big]\bigg]\implies \big(\mathfrak{I}_{t+1} = \mathfrak{I}_{t}\big).
\]

\end{enumerate}

\begin{remark}
    These rules \ref{D1}–\ref{D2} capture the idealized assumption that the agents can and will reorganize any time it is possible to make someone better off without making anyone worse off, as if any coalition of agents that could benefit at no one’s expense instantly forms and executes the required deal. This might sound unrealistic in practice – it assumes costless, fully informed multilateral bargaining – but it represents the theoretical limit of what frictionless negotiation could achieve. 
\end{remark}

We now state and prove the main theorem under the revised rules, which essentially says: If property rights are well-defined, there are no transaction costs, and if agents are able to execute any exchange that is beneficial to at least one and harmful to none, then the allocation of resources will converge to a Pareto optimal state and remain there. This is a formal version of the Coase theorem.

\begin{theorem}[The Coase Theorem]
    Every initial distribution of resources converges to some ultimate distribution that is Pareto optimal, i.e.,
\[
    \forall \, \mathfrak{I}_0 \in \mathcal{P}(R)^n \, \exists \, t \in \mathbb{N} \, \forall \, k \in \mathbb{N} \Bigg[ \big[ \mathfrak{I}_t = \mathfrak{I}_{t+k} \big] \, \land \, \forall \, \mathfrak{K} \neq \mathfrak{I}_{t+k} \in \mathcal{P}(R)^n \,\]
    \[ \bigg[\, \exists \, a_q \in \mathbb{A} \,\big[\hat{\mathfrak{I}}_{t+k}(a_q) \prec \hat{\mathfrak{K}}(a_q)\big] \implies \exists \, a_z \neq a_q \in \mathbb{A} \, \big[\hat{\mathfrak{K}}(a_z) \prec \hat{\mathfrak{I}}_{t+k}(a_z)\big]\bigg] \Bigg].
\]
\end{theorem}
\begin{proof}
  Suppose, to the contrary, that there exists an initial distribution 
\(\mathfrak{I}_0\) that either fails to converge to some ultimate distribution or fails to converge to an ultimate distribution that is Pareto optimal. There are two cases: either \(\mathfrak{I}_0\) is Pareto optimal, or it is not. If \(\mathfrak{I}_0\) is Pareto optimal, then so is \(\mathfrak{I}_1\); by \ref{D2}, every subsequent distribution must be Pareto optimal, since they are all identical to \(\mathfrak{I}_0\). On the other hand, if \(\mathfrak{I}_0\) is not Pareto optimal, then there must exist an alternative distribution \(\mathfrak{K}\) for which there is an agent \(a_q\) who strictly prefers \(\mathfrak{K}\) over \(\mathfrak{I}_0\) without any other agent strictly preferring \(\mathfrak{I}_0\) over \(\mathfrak{K}\). There cannot be a situation where every subsequent distribution remains suboptimal. Suppose, to the contrary, that every subsequent distribution is suboptimal; then for every subsequent distribution, it is always possible to make someone better off without there being someone who is worse off, thereby contradicting the assumption of finite resources. It follows that the distribution must converge, and by \ref{D2}, only does so \textit{when} the outcome is Pareto optimal. 
\end{proof}

It might appear that by building Pareto optimality into the very rule \ref{D1}, we have proven the Coase theorem almost by definition – indeed, we have essentially assumed that agents have the capability to directly implement any Pareto improvement. This assumption is admittedly strong, but it isolates the core conceptual requirement for Coase’s vision to hold: agents must be able to coordinate exchanges beyond simple pairwise trades whenever such coordination can make everyone better off. In reality, achieving such coordination may involve bargaining institutions, intermediaries, or enforceable contracts among multiple parties. However, Coase’s theorem is an ideal benchmark – it presumes those coordination issues away by assuming zero transaction costs and perfectly defined rights. Our formalism with ideal exchanges is an accurate mirror of that presumption.

\section{Philosophical \& Concluding Remarks}\label{Philosophical Remarks}

What, then, is the philosophical takeaway? One observation is that Coase’s theorem is not a trivial truth of logic or economics; it holds only under specific conditions about the process of interaction. By formalizing the theorem, we made explicit some often implicit assumptions. The notion of ``well-defined property rights'' and ``no transaction costs'' were modeled by saying every resource is always owned and, initially in Section 4, by limiting trades to those both parties strictly benefit from. These were not enough to guarantee efficiency. We had to strengthen ``no transaction costs'' to an extreme degree – allowing costless multilateral bargaining – to get the efficiency result. This reveals that the Coase theorem’s truth hinges on assuming away not just pricing or negotiation costs, but also the complexities of coalition formation. 

Our formal model also provides a new perspective on the role of entrepreneurial foresight and coordination. In the philosophical literature on economics, much is made of the idea of agents as optimization machines who will find any available improvement. Coase’s view of bargaining could be seen as assuming agents have the farsightedness and perhaps the entrepreneurial spirit to realize complex exchanges. In our terms, an ideal exchange requires that agents are aware of the collective benefit and can trust that no one will be harmed by participating. One might say: if agents truly recognize a three-way trade is to all of their benefit, they have a strong incentive to implement it, perhaps through a sequence of trades or a multilateral contract. Thus, whether we extend this presumption of foresight and entrepreneurial courage to any market participants is ultimately an empirical or at least a pragmatic question. Coase’s theorem does not hold come what may; it holds if people can bargain in an idealized way.

Finally, stepping back to the broader motivation: we used set theory and logic to axiomatize a slice of economic theory. Was this effort worthwhile? \textbf{We believe it was}, for several reasons. It forced us to clarify subtle concepts like what exactly constitutes a ``transaction cost'' in a model. Here, anything that prevents a Pareto-improving move, including the inability of three people to bargain together, can be interpreted as a transaction cost or impediment. It also allowed us to uncover assumptions that might otherwise be overlooked – for example, the standard textbook treatment with an Edgeworth box implicitly assumes that the contract curve will be reached, but fails to spell out an effectively calculable procedure. 

By constructing one and seeing it fail, we identified that a gap exists for multi-party exchanges. This is a case where the philosophy of mathematics – in the sense of scrutinizing the foundations and logical structure of a theory – contributes to economic understanding. We have not proven anything that could not be intuited (indeed, our formal results echo known insights from economics), but we have provided a formal verification and clarity. Recalling Cohen, the aim is not to discover new truths per se, but to precisely articulate the rules and assumptions in use. In our case, we articulated the ``rules of the game'' that lead to Coasean outcomes or their absence.

In conclusion, the Coase theorem can be affirmed as a logically valid result if we assume ideal conditions for bargaining. The presence of transaction costs, broadly interpreted to include any impediment to comprehensive agreements, can invalidate the theorem. None of this undermines Coase’s fundamental insight that when people can bargain freely, they will tend to remove inefficiencies – but it does remind us that ``bargaining freely'' is doing a lot of work in that sentence. \textbf{The formalization presented helps pinpoint what exactly ``freely'' must entail.} It is our hope that such formal exercises, while abstract, can inform practical discussions by illuminating why certain real-world situations deviate from the ideal and what institutional arrangements might be needed to approach the efficient outcomes that the Coase theorem predicts under its stringent conditions.

\section*{Acknowledgments} 

The author is grateful to Noah Wang (University of St Andrews) and Professor Benedict Eastaugh (University of Warwick) for their lasting inspiration and guidance. The author also thanks Professor Joel David Hamkins (University of Notre Dame) for his helpful comments. 

\section*{Declarations}

The author declares no conflicts of interest. No funding was received for this work. This work is purely theoretical and does not involve the generation of any data sets.

\end{document}